\documentclass[a4paper,english]{lipics-v2019}
\usepackage{nicefrac}
\usepackage{microtype}
\usepackage{todonotes}
\usepackage{xcolor}
\usepackage{xspace}
\usepackage{amsthm, thm-restate}
\usepackage{cleveref}

\theoremstyle{plain}
\newtheorem{observation}{Observation}

\newtheorem{lem}[observation]{Lemma}
\newenvironment{proof-sketch}{%
\renewcommand{\proofname}{Proof sketch}\proof}{\endproof}

\newcommand{\old}[1]{{}}

\newcommand{\mscA}{\textsc{a-msc}\xspace}
\newcommand{\mscG}{\textsc{msc}\xspace}

\crefname{figure}{Figure}{Figures}
\crefname{theorem}{Theorem}{Theorems}
\crefname{lemma}{Lemma}{Lemmas}
\crefname{lem}{Lemma}{Lemmas}
\crefname{observation}{Observation}{Observations}
\crefname{corollary}{Corollary}{Corollaries}
\crefname{section}{Section}{Sections}
\relatedversion{An extended abstract based on the content of this paper has been accepted for inclusion in the Symposium on Computational Geometry (SoCG) 2020 \cite{thesocgpaper}.}
\graphicspath{{./figures/}}

\hideLIPIcs
\nolinenumbers

\title{Minimum Scan Cover \newline with Angular Transition Costs}

\titlerunning{Minimum Scan Cover}

\author{S\'{a}ndor P. Fekete}{Department of Computer Science, TU Braunschweig, Germany}{s.fekete@tu-bs.de}{https://orcid.org/0000-0002-9062-4241}{}
\author{Linda Kleist}{Department of Computer Science, TU Braunschweig, Germany}{l.kleist@tu-bs.de}{https://orcid.org/0000-0002-3786-916X}{}
\author{Dominik Krupke}{Department of Computer Science, TU Braunschweig, Germany}{d.krupke@tu-bs.de}{https://orcid.org/0000-0003-1573-3496}{}

\authorrunning{S. P. Fekete and L. Kleist and D. Krupke}
\Copyright{S\'{a}ndor P. Fekete and Linda Kleist and Dominik Krupke}



\acknowledgements{We thank Phillip Keldenich, Irina Kostitsyna, Christian Rieck, and Arne Schmidt for helpful algorithmic discourse, 
Kenny Cheung (NASA) and Christian Schurig (European Space Agency) for joint work on intersatellite communication, and
Karl-Heinz Gla{\ss}meier for discussions of astrophysical aspects.}

\begin{document}

\maketitle

\begin{abstract}
We provide a comprehensive study of a natural geometric optimization problem motivated
by questions in the context of satellite communication and astrophysics. In
the problem {\sc Minimum Scan Cover with Angular Costs} (\mscG), we are given a graph $G$
that is embedded in Euclidean space. The edges of $G$ need to be \emph{scanned},
i.e., probed from both of their vertices. In order to scan their edge, two vertices need to face each other; changing the heading of a vertex
takes some time proportional to the corresponding turn angle. Our goal
is to minimize the time until all scans are completed, i.e., to compute a
schedule of minimum makespan.

We show that \mscG is closely related to both graph coloring and the
minimum (directed and undirected) cut cover problem; in particular,
we show that the minimum scan time for instances in 1D and 2D lies 
in $\Theta(\log \chi (G))$, while for 3D the minimum scan time is not 
upper bounded by $\chi (G)$. We use this relationship
to prove that the existence of a constant-factor approximation implies
$P=NP$, even for one-dimensional instances. In 2D, we show that it 
is NP-hard to approximate a minimum scan cover within less than a 
factor of $\nicefrac{3}{2}$, even for bipartite graphs; conversely,
we present a $\nicefrac{9}{2}$-approximation algorithm for this scenario.
Generally, we give an $O(c)$-approximation for $k$-colored graphs with $k\leq \chi(G)^c$.
For general metric cost functions, we provide approximation algorithms whose performance
guarantee depend on the arboricity of the graph.

\end{abstract}

\newpage
\section{Introduction}
Many problems of geometric optimization arise from questions of communication,
where different locations need to be connected.
For physical networks, the cost of a connection corresponds to
the geometric distance between the involved vertices, e.g., the length
of an electro-optic link. Often wireless transmissions may be used
instead; however, for ultra-long distances such as in space, this requires
focused transmission, e.g., by communication partners
facing each other with directional, paraboloid antennas or laser beams.
This makes it
impossible to exchange information with multiple partners at once; moreover,
a change of communication partner requires a change of heading, which is
costly in the context of space missions with limited resources,
making it worthwhile to invest in a good schedule.

With the advent of satellite swarms of ever-growing size, problems of 
this type are of increasingly practical importance for 
ensuring communication between spacecraft. They also come 
into play when astro- and geophysical measurements are to be performed,
in which groups of spacecraft can determine physical quantities not just
at their current locations, but also along their common line of sight.

\begin{figure}[htb]
        \centering
        \includegraphics[width=.5\textwidth]{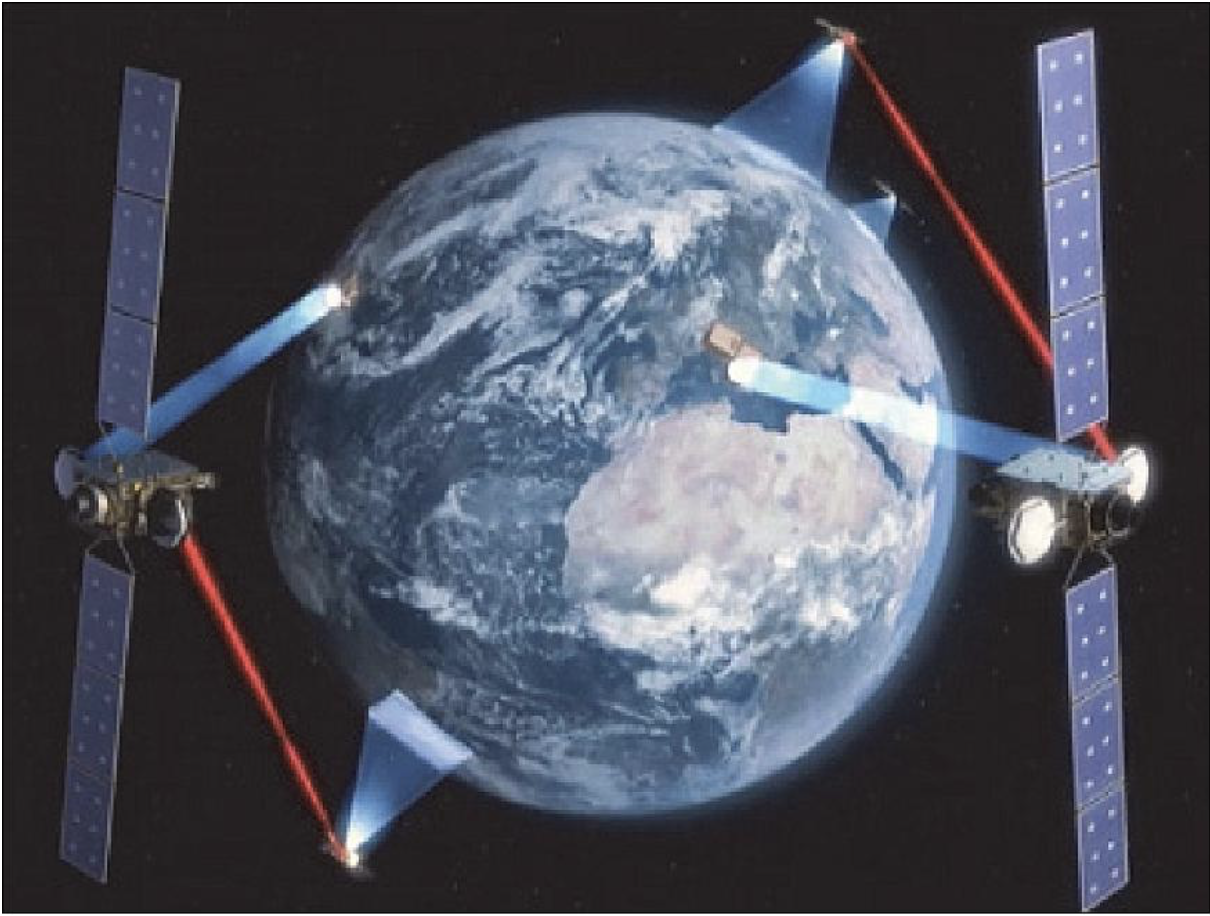}
        \caption{Artist's rendition of the European Data Relay Satellite constellation architecture. Note the intersatellite links shown in red. (Image credit: ESA)}
        \label{fig:esa}
\end{figure}

We consider an optimization problem arising from this
context: How can we schedule a given set of intersatellite
communications, such that the overall timetable is as efficient as possible?
In particular, we study
the question of a {\sc Minimum Scan Cover with Angular Costs} (\mscG), in which
we need to establish a collection of connections between a given set of
locations, described by a graph $G=(V,E)$ that is embedded in space.
For any connection (or \emph{scan}) of an edge,
the two involved vertices need to face each other; changing the heading of a vertex
to cover a different connection takes an amount of time proportional to the 
corresponding turn angle. Our goal is to minimize the time until all 
tasks are completed, i.e., compute a geometric schedule of minimum makespan.

In this paper, we provide a comprehensive study of this problem.
We show that \mscG is closely related to both graph coloring and the
minimum (directed and undirected) cut cover problem. We also provide a
number of hardness results and approximation algorithms for a variety
of geometric scenarios; see Section~\ref{overview} for an overview.

\subsection{Problem Definition: Minimum Scan Cover}

In the \emph{abstract version} of {\sc Minimum Scan Cover}, denoted by \mscA, 
we are given a simple graph $G=(V,E)$ and a
metric cost function $\alpha:\{( e, e')\in E\times E\mid
e\cap e'\neq \emptyset\}\to\mathbb R^+$ that describes the cost for switching between
two incident edges $uv, vw$ for the common vertex $v$.
A \emph{scan cover} is
an assignment $S:E\to \mathbb{R}^+$, such that for every vertex~$v$ and every pair of
incident edges $uv$ and $vw$ it holds that 
\[|S(uv)-S(vw)|\geq \alpha(uv, vw).\]
This condition provides sufficient time for the vertices to 
face each other.
We seek a scan cover that minimizes the \emph{scan
time} $\max_{e\in E}S(e)$. Note that \mscA generalizes the Path-TSP (see \cref{obs:starsTSP}), 
so the problem becomes intractable if the cost function $\alpha$ is not metric.

Given the practical motivation, our main focus is the
\textsc{(geometric) Minimum Scan Cover Problem} (\mscG), for which every vertex $v$ corresponds to a point
in $\mathbb R^d$, $d\in\{1,2,3\}$, the turn cost $\alpha(uv,vw)$ of $v$
from $uv$ and $vw$ is the (smaller) angle at $v$ between 
the segments $uv$ and $vw$. 
\cref{fig:introExample} illustrates a minimum scan cover for a point set in the plane that can be scanned in $300^\circ$ with eleven discrete time steps.

\pgfkeys{/pgf/number  format/.cd,fixed,precision=0}
\begin{figure}[htb]
\vspace{6pt}
\centering
\includegraphics[page=1]{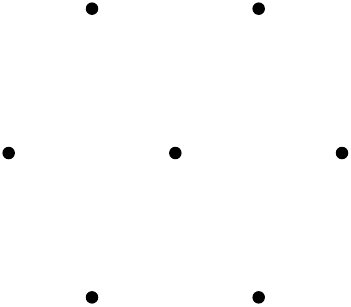}\hfil
\includegraphics[page=12]{./figures/IntroExample6gon.pdf}
\vspace{12pt}

\noindent
\foreach \x in {2,3,4,5,6}{
	\makebox(20,60)[r]{\pgfmathparse{(\x-2)*30}\pgfmathprintnumber{\pgfmathresult}$^\circ$}
\includegraphics[page=\x,scale=.4]{./figures/IntroExample6gon.pdf}\hfill
}

\vspace{12pt}
\noindent
\foreach \x in {7,8}{
\makebox(20,60)[r]{\pgfmathparse{(\x-2)*30}\pgfmathprintnumber{\pgfmathresult}$^\circ$}	\includegraphics[page=\x,scale=.4]{./figures/IntroExample6gon.pdf}\hfill
}
\foreach \x in {9,10,11}{
\makebox(20,60)[r]{\pgfmathparse{(\x-1)*30}\pgfmathprintnumber{\pgfmathresult}$^\circ$}	\includegraphics[page=\x,scale=.4]{./figures/IntroExample6gon.pdf}\hfill
\vspace{3pt}
}
  \caption{\textbf{(Left)} A set of seven points in $\mathbb R^2$, for which the complete graph $K_7$ needs to be scanned.
  \textbf{(Bottom)} A sequence of edge scans; note how some edges can be scanned in parallel. 
 }
\label{fig:introExample}
\end{figure}

In fact, a scan cover is completely determined by an edge order:
For each edge sequence $e_1,\dots,e_m$, the best scan cover that scans the edges in this order can be computed by 
\[S(e_1)=0 \text{ and }S(e_i)=\max\{S(e_j)+\alpha(e_i,e_j)|j: j<i,e_i\cap e_j\neq \emptyset )\} \text{ for } i>1.\]

In some settings, we may be given an initial heading of each vertex.
However, the cost of changing from the initial heading to any other is usually negligible compared to the cost of the remaining schedule.
In fact, any $C$-approximation without initial headings yields a
$(C+1)$-approximation for the variant with initial headings, as turning from
the initial heading to any edge via a smallest angle is not more expensive
than the minimum makespan.

\subsection{Overview of Results and Organization}\label{overview}
In \cref{sec:1D}, we show that the \mscG in 1D corresponds to a minimum
directed cut cover and has a strong correlation to the chromatic number.  We
provide an improved upper bound of $\lceil \log_2 \chi (G)+\frac 1 2 \log_2
\log_2 \chi (G)+1\rceil$ (\cref{thm:boundLine,cor:directedCut}) for the minimum
directed cut cover number, which is essentially tight in general; even for directed acyclic graphs corresponding to
minimum scan covers (\cref{lem:1dtight}) this is in the
right order of magnitude. 
This implies that, unless $P=NP$, there exists no constant-factor approximation even in 1D (\cref{thm:noapx}). Nevertheless, we show that instances in which the underlying graphs are bipartite or complete graphs can be solved in polynomial time (\cref{obs:1dbipopt,obs:1dcomplopt}).

In \cref{sec:2D}, we consider the problem in 2D and show that it is NP-hard to approximate minimum scan covers of bipartite graphs better than $\nicefrac{3}{2}$ (\cref{thm:hardnessBIP}). Furthermore, we provide absolute and relative bounds. On the one hand, every bipartite graph in 2D has a scan cover of length $360^\circ$ (\cref{lemma:scan_points_splited_by_line}). On the other hand, we present a $\nicefrac{9}{2}$-approximation algorithm (\cref{thm:approxBip}). More generally, 
 we present an $O(c)$-approximation for a $k$-colored graph with $k\leq\chi(G)^c$ (\cref{cor:approxCol}). This has immediate consequences for several interesting graph classes, e.g., the scan time of graphs in 1D and 2D lies in $\Theta(\log_2 {\chi(G)})$ and there exist constant factor approximations (\cref{cor:specialCases,cor:complete2D}).

In \cref{sec:3D}, we consider \mscG in 3D and the abstract version \mscA. We show that in contrast to 2D, the length of a minimum scan cover in 3D may exceed $O(\log_2 n)$ (\cref{obs:noboundbychrom}). 
Complementary to the fact that \mscA for stars is equivalent to path-TSP and thus NP-hard, we provide a $2.5$-approximation of \mscA for trees (\cref{thm:3DTree}).
This yields an $O(A)$-approximation for every graph with arboricity $A$ (\cref{thm:amscArboricityApprox}).

\subsection{Related Work}
The use of directional antennas has introduced a number of geometric questions.
Carmi et al.~\cite{carmi2011connectivity} study
the $\alpha$-MST problem, which arises from finding orientations of directional antennas with $\alpha$-cones, such 
that the connectivity graph yields a spanning tree of minimum weight, based on bidirectional communication.
They prove that for $\alpha<\nicefrac \pi 3$, a solution may not exist, while $\alpha\geq\nicefrac \pi 3$ always suffices.
See Aschner and Katz~\cite{aschner2017bounded} for more recent hardness proofs and constant-factor 
approximations for some specific values of $\alpha$.

Many other geometric optimization problems deal with turn
cost. Arkin et al.~\cite{arkin2001optimal,arkin2005optimal} 
show hardness
of finding an optimal milling tour with turn cost, even in relatively constrained settings,
and give a $2.5$-approximation algorithm for obtaining a cycle cover,
yielding a $3.75$-approximation algorithm for tours.
The complexity of finding an optimal cycle cover in a 2-dimensional grid
graph was stated as \emph{Problem~{53}} in \emph{The Open Problems
Project}~\cite{openproblemproject} and shown
to be NP-complete in~\cite{fk-ctcct-19}, which also provides constant-factor
approximations; practical methods and results are given in~\cite{ALENEX19},
and visualized in the video~\cite{bdf+-zzmmd-17}.

Finding a fastest roundtrip for a set of points in the plane for which the
travel time depends only on the turn cost is called the \emph{Angular Metric
Traveling Salesman Problem}. Aggarwal et al.~\cite{aggarwal2000angular} prove
hardness and provide an $O(\log n)$ approximation algorithm for cycle covers
and tours that works even for distance costs and higher dimensions. 
For the abstract version on graphs in which ``turns'' correspond to weighted changes between edges, 
Fellows et al.~\cite{fellows2009abstract} show that the problem is fixed-parameter tractable in the number of turns, the treewidth, and the maximum degree.
Fekete and Woeginger~\cite{fekete1997angle} consider the problem of connecting a set of points
by a tour in which the angles of successive edges are
constrained.
\pagebreak

Our paper also draws connections to other graph optimization
problems.
In particular, for each point in time, the set of scanned edges induces a bipartite graph.
Therefore, one approach for scanning all edges of the given graph is to partition it into a small
number of bipartite graphs, each corresponding to the set of edges separated by
the \emph{cut} induced by a partition of vertices into two non-trivial sets. 
This problem is also known as the \textsc{Minimum Cut Cover Problem}:
Find the minimum number of cuts to cover all edges of a graph. Loulou~\cite{Loulou:1992:MCC:2307951.2308121} shows that for complete graphs,
an optimal solution consists of $\lceil \log_2 |V| \rceil$ cuts. Motwani and
Naor~\cite{Motwani:1994:EAC:891890} prove that, unless $P=NP$, the problem on general graphs
is not approximable within $1.5$ of the optimum, or $OPT+\varepsilon \log |V|$
for some $\varepsilon>0$ in absolute terms, due to a direct relationship with 
graph coloring. Hoshino~\cite{DBLP:journals/endm/Hoshino11} considers 
practical methods based on integer programming and heuristics for cut covers. Chuzhoy and
Khanna~\cite{Chuzhoy:2006:HCP:1132516.1132593} show that 
the directed version of 
covering a directed graph by the minimum number of directed cuts is also an NP-hard problem. 

On the application side, Korth et al.~\cite{korth2002particle} describe the use of tomography (i.e., 
determining physical phenomena by measuring aggregated effects along a ray 
between two sensors) in the context of astrophysics. Using multiple sensors 
(e.g., satellites) for performing efficient measurements is one
of the motivations for the algorithmic work in this paper.
Scheduling satellite communication has received a growing amount of attention,
corresponding to the increasing size of satellite swarms.
See Krupke et al.~\cite{krupke2019automated} for a recent overview.

In the context of scheduling,
Allahverdi et al.~\cite{allahverdi2015third,allahverdi1999review,allahverdi2008survey} provide a nice and comprehensive survey on scheduling variants with sequence-dependent setup costs.
Sotskov et al.~\cite{sotskov2001mixed} consider a scheduling variant that can directly be expressed as vertex coloring.

\section{One-Dimensional Point Sets}\label{sec:1D}

In the one-dimensional case, all vertices lie on a single line $L$.
Therefore, an instance can be described by a graph $G=(V,E)$ and a total order of the
vertices $<_L$ on $L$. We assume this line to be horizontal, so 
vertices face either left or right when scanning an edge. 
Moreover, scan times can be restricted to discrete multiples of $180^\circ$. 
This allows us to encode the headings of a vertex $v$ at these time steps by
a 0-1-vector $s(v)$, where a right heading is denoted by $0$, and a left one
by $1$; we denote by $s_i(v)$ the $i$th bit of $s(v)$. Then a scan cover with $N$ steps
of $(G,<_L)$ is an assignment $s\colon V\to \{0,1\}^N$, such that for every
edge $uv\in E$, $u <_L v$, there exists an index $i\in [N]$ with
$s_i(u)=0$ and $s_i(v)=1$. The value of such a scan cover is clearly $180^\circ (N-1)$.
For an example, consider \cref{fig:Line}. 

\begin{figure}[htb]
  \centering
  \includegraphics[page=3]{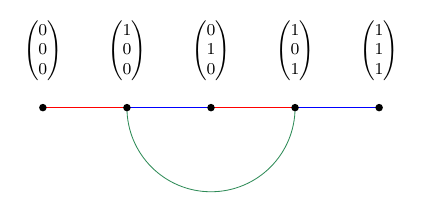}
  \caption{This instance can be scanned in three steps. However, two steps are not sufficient because the edges of a monotone path would need to  be scanned in alternating time steps; making it impossible to scan the green edge.
  }
  \label{fig:Line}
\end{figure}

\pagebreak

\subsection{Bounds Based on Chromatic Number and Cut Cover Number}

In the following, we establish a strong relationship between the length of a \mscG in 1D
and the chromatic number $\chi(G)$, which is closely linked to the cut cover number
$c(G)$ of the involved graph $G=(V,E)$, i.e., the size of a smallest
partition of the edge set into bipartite graphs. Motwani and
Naor~\cite{Motwani:1994:EAC:891890} show that 
\[c(G)=\lceil \log_2 \chi (G) \rceil.\]

Because the scanned edges in each time step form a bipartite graph,
a scan cover induces a cut cover. However, the resulting bipartite
graphs have the additional property that for each vertex all neighbors are
either smaller or larger with respect to $<_L$. Thus, not every cut cover
corresponds to a scan cover. 
However, scan covers correspond to \emph{directed}
cut covers of the directed graph, induced by orienting the edges from left to right. 
Watanabe et al.~\cite{directedCutCover} 
bound the directed cut cover number $\vec c (G)$ of a directed graph~$G$:
\[\vec c (G)\leq \lceil \log_2 \chi (G) \rceil+\lceil \log_2 \lceil \log_2 \chi (G)+1\rceil\rceil \]
We improve this bound by showing an upper bound for the size of a smallest scan
cover in terms of the chromatic number (and the cut cover number); this bound is best possible
for the directed cut cover number as we explain later.

\begin{theorem}\label{thm:boundLine}
  For every graph $G$ with $\chi(G)\geq 2$ and every ordering $<_L$ of
  the vertices, there exists a scan cover of $(G,<_L)$ with $N$ steps such
  that
  \begin{equation*}
    N \leq \lceil \log_2 \chi (G)+\frac 1 2 \log_2 \log_2 \chi (G)+1\rceil
  \end{equation*}
\end{theorem}

\begin{proof}
  Consider a coloring of $G$ with $C:=\chi(G)$ colors and choose an $N$
  large enough such that $C\leq \binom{N}{\lfloor N/2\rfloor}$. For
  $k:=\lfloor\frac{N}{2}\rfloor$, we consider the set of vectors
  $\{0,1\}^N_k$ of length $N$ with exactly $k$ many $1$'s. We define a scan
  cover $s:V\to\{0,1\}^N_k$, such that for all vertices of the same color, we
  assign the same vector, while vertices of different color obtain different
  vectors. Such an assignment exists, because the number of vectors, i.e.,
  $\binom{N}{\lfloor N/2\rfloor}$, is at least as large as the number of colors.

  To see that $s$ is a scan cover, consider a fixed but arbitrary edge
  $uv$ of $G$. Because the vectors $s(u)$ and $s(v)$ differ but have the same
  number of 1's, they are \emph{incomparable}, i.e., there exist $i$ and $j$ such that $s_i(u)=0$, $s_i(v)=1$ and
  $s_j(u)=1, s_j(v)=0$. Therefore, depending on the ordering of $u$ and $v$ on $L$, the edge $uv$ is either scanned in step $i$ or $j$. 

  It remains to show that defining $N:= \lceil\log_2 C+\frac 1 2 \log_2 \log_2 C+1\rceil$ satisfies $C\leq \binom{N}{\lfloor N/2\rfloor}$.
  By a variant of Stirling's formula \cite{stirling}, it holds that 

  \begin{equation*}
    e^{1/(12n+1)}\leq \frac{n!}{\sqrt{2\pi n}(n/e)^n}\leq e^{1/{(12n)}}.
  \end{equation*}
  This implies that 
  $\binom{N}{\lfloor N/2\rfloor}\geq\sqrt{\frac{2}{\pi N}}\cdot 2^N\cdot e^{\frac{-1}{4N-1}},$
  so it suffices to guarantee
  \begin{equation*}
    C\leq  \sqrt{\frac{2}{\pi N}}\cdot 2^N\cdot e^{\frac{-1}{4N-1}}
    \iff \log_2 C\leq N +\frac{1}{2} (1-\log_2 \pi -\log_2 N)-\frac{1}{4N-1}\log_2 e.
  \end{equation*}
  If $C\geq 3$, this holds for $N= \lceil\log_2 C+\frac 1 2 \log_2 \log_2
  C+1\rceil\geq 3$ ; in case of $C=2$, it holds that $N= \lceil\log_2 C+\frac 1 2
  \log_2 \log_2 C+1\rceil=2$, and thus $C\leq \binom{N}{\lfloor N/2\rfloor}$.
\end{proof}

Note that the assigned vectors in the proof of \cref{thm:boundLine} are pairwise incomparable. Therefore, such an assignment yields a directed cut cover for all edge directions and thus a general bound on the directed cut cover number.

\begin{corollary}\label{cor:directedCut}
  For every directed graph $G$, the directed cut cover number is bounded by 
  $$\vec c(G)\leq \lceil \log_2 \chi (G)+\frac 1 2 \log_2 \log_2 \chi (G)+1\rceil.$$
\end{corollary}
In fact, the bound in \cref{cor:directedCut} is best possible for general directed graphs, because a cut cover of the complete bidirected graph corresponds to an assignment of pairwise incomparable vectors (and Sperner's theorem asserts that the used set of vectors is maximal).

\cref{fig:Line} illustrates an example of a graph $G$ and an ordering $<_L$ showing that the bound of \cref{thm:boundLine,cor:directedCut} is also tight for some (directed acyclic) graphs with $\chi(G)=3$.
In the following, we show a general lower bound for our more special setting. 
\begin{lem}\label{lem:1dtight}
  For every $C$, there exists a graph $G$ and an ordering $<_L$ such that $\chi(G)>C$ and the number $N$ of steps in every scan cover of $(G,<_L)$ is at least
  \begin{equation*}
    N\geq \lceil\log_2 \chi(G)+\frac 1 4 \log_2 \log_2 \chi(G)\rceil.
  \end{equation*}
\end{lem}
\begin{proof}
  Let $\ell\geq  4$ be an integer divisible by 4 and $n:=2^{\ell}$ such that $2^n>C$.
  We consider the Turan graph~$G$ on $n2^n$ vertices partitioned into
  $2^n$ independent sets of size $n$. Because $G$ is a complete $2^n$-partite
  graph, it holds that $\chi(G)=2^n$. We place the vertices on the line, such
  that for a fixed $\{1,\dots,2^n\}$-coloring of~$G$, there exist $n$ disjoint intervals
  in which the colors appear in the order $1,\dots,2^n$. For an illustration
  consider \cref{fig:LineLowerBound}.

  \begin{figure}[htb]
    \centering
    \includegraphics[page=1]{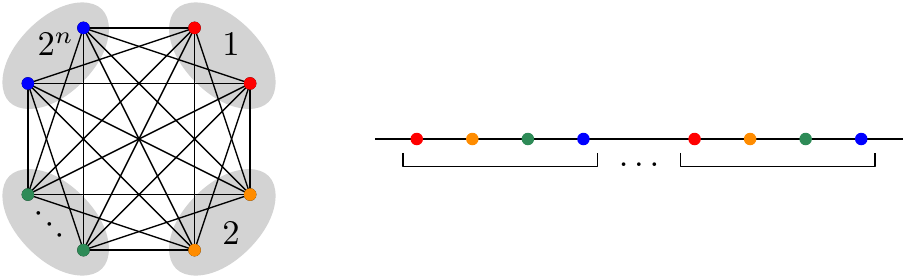}
    \caption{Illustration of $G$ and the ordering $<_L$ of the vertices on $L$ for  $n=2$ ($\ell=1$).}
    \label{fig:LineLowerBound}
  \end{figure}

  For a contradiction, suppose that there exists a scan cover $s:V\to \{0,1\}^k$ of $(G,<_L)$ with 
  $k:=\lceil\log_2 \chi(G)+\frac 1 4 \log_2 \log_2 \chi(G)\rceil-1
  =n+\frac{\ell}{4}-1$ steps.
  Thus, the number of different vectors is $2^k=2^{n-1} n^{1/4}$.

  Let $t$ denote the number of different color classes in which some vector is
  used at least $n^{3/4}$ times.
  We show that $t\geq \frac{1}{2}2^n$. Clearly,
  each vector may only appear in one color class, i.e., the color classes induce
  a partition of the set of vectors.
  Consider the $2^n-t$ color classes (and their assigned vectors) in which no vector is used $n^{3/4}$ times.
  Let $\delta$ denote the average usage of vectors in these classes.
  Note that $\delta$ is lower bounded by the ratio of the number of vertices, namely $(2^n-t)n$, and the maximum number of remaining vectors, namely $2^k-t$. Consequently, $\delta\geq \frac{n2^n-tn}{2^k-t}$.
  Moreover, $\delta<n^{3/4}$, because otherwise there exists a further color class for which
  some vector appears at least $n^{3/4}$ times. Therefore, we obtain the following
  chain of implications: \[\delta<n^{3/4} \implies \frac{n2^n-tn}{2^k-t} < n^{3/4} \iff
  t> 2^n\cdot\frac{1}{2(1-n^{-1/4})}  \implies t > \frac{1}{2}2^n\]

  For each of these $t$ color classes, we choose a vector with a maximal number of appearances and introduce an interval on $L$ from the first to the last occurrence.
  By the ordering of the vertices, every two vertices of the same color have a distance of at least $2^n$, and hence the interval spans at least $d=2^n n^{3/4}$ vertices.
  On average, every vertex is contained in the following number of intervals
  \[\frac{t\cdot d}{|V|}\geq\frac{\frac{1}{2}2^n\cdot2^nn^{3/4}}{n2^n}=\frac{2^n}{2 n^{1/4}}=2^{n-1}n^{-1/4}.\]
  By the pigeonhole principle, there exists a set $S$ of at least $2^{n-1}n^{-\nicefrac 1 4}$ vectors with mutually intersecting intervals.
  We claim that any two vectors $a$ and $b$ of $S$ are pairwise \emph{incomparable}, i.e., there exist two indices $i,j$ such that $a_i=0,b_i=1$ and $a_j=1,b_j=0$:
  Because the intervals intersect, among the four occurrences of $a$ and $b$ on $<_L$, there exist three such that they appear alternating.
  To scan the corresponding edges, the vectors must be incomparable.
  Thus, there must exist $2^{n-1}n^{-\nicefrac 1 4}$ pairwise incomparable vectors.

  However, by Sperner's theorem, every set of vectors of length $k$ contains at most $\binom{k}{\lfloor k/2\rfloor}$ pairwise incomparable vectors and 
  \[\binom{k}{\lfloor k/2\rfloor}\leq \sqrt{\frac{2}{k \pi}}2^k(1+\frac{1}{11})\leq2^k \frac{ 1}{\sqrt{k}}.\]
  It remains to show that the number of necessary incomparable vectors exceeds this:
  \begin{align*}
    2^k \cdot \frac{ 1}{\sqrt{k}}<\frac{2^n}{2 n^{1/4}}
    \iff  n<k
  \end{align*}
  This holds for $\ell>4$ and yields a contradiction. For $\ell= 4$ it holds that $k=n$. Thus, each color class has a unique vector, all of which need to be incomparable, a contradiction.
\end{proof}

\subsection{No Constant-Factor Approximation in 1D}

\cref{thm:boundLine} implies the following.
\begin{lem}\label{thm:coloring}
  A $C$-approximation algorithm for \mscG 
  implies a polynomial-time algorithm for computing a coloring of graph $G$, 
  $k:=\chi(G)$, with $4^C\cdot k^C\cdot \sqrt{\log_2(k)} ^C$ colors. 
\end{lem}
\begin{proof}
  Let $\ell^*$ denote the length of a minimum scan cover of $G$.  Then a $C$-approximation algorithm computes a scan cover of length $\ell\leq C \cdot \ell^*$. 
  \cref{thm:boundLine} implies that $C \cdot \ell^*\leq C\cdot \lceil \log_2 k+\frac 1 2 \log_2 \log_2 k+1\rceil$, yielding a coloring with $2^\ell$ colors.
  Thus, 
  \[
    2^\ell
    \leq 2^{C(\lceil \log_2 k+\frac 1 2 \log_2 \log_2 k+1\rceil)}
    \leq 2^{C\cdot \log_2 k}\cdot 2^{\frac{1}{2}\cdot C\cdot \log_2\log_2 k}\cdot 2^{2C}
    \leq 4^C\cdot k^C\cdot \sqrt{\log_2(k)} ^C.\qedhere
  \]
\end{proof}

\begin{theorem}
  \label{thm:noapx}
  Even in 1D, a $C$-approximation for \mscG for any $C\geq 1$ implies $P=NP$.
\end{theorem}
\begin{proof}
  Suppose there is a $C$-approximation for some constant $C>1$. By
  \cref{thm:coloring}, a $C$-approximation of \mscG in 1D implies that there
  is a polynomial-time algorithm for finding for every $k$-colorable graph $G$ a
  coloring with $4^C\cdot k^C\cdot \sqrt{\log_2(k)} ^C$ colors.
  Khot~\cite{khot2001improved} showed that, for sufficiently large $k$, it is NP-hard to color a $k$-colorable graph with at most $k^{\log_2(k)/25}$ colors.
  However, for every $C$ we can find a $k$ such that $4^C\cdot k^C\cdot \sqrt{\log_2(k)} ^C <k^{\log_2(k)/25}$. This yields a
  polynomial-time algorithm for an $NP$-hard problem, implying that $P=NP$.
\end{proof}

\subsection{Polynomially Solvable Cases}

Even though there is no constant-factor approximation in general, we would like to note that bipartite and complete graphs in 1D can be solved in polynomial time. 

\begin{observation}\label{obs:1dbipopt}
  For instances of \mscG in 1D for which the underlying graph $G$ is bipartite, there exists a polynomial-time algorithm for computing an optimal scan cover.
\end{observation}
\begin{proof}
  We assume that $\chi(G)=2$, otherwise there is no edge to scan. If for
  every vertex, all its neighbors lie either before or after it, $G$ can be scanned within one
  step, which is clearly optimal.
  Otherwise, every scan cover needs at least two
  steps. By \cref{thm:boundLine}, there exists a scan cover with $2$ steps.
  Because bipartite graphs can be colored in polynomial time, the proof of
  \cref{thm:boundLine} provides a scan cover.
\end{proof}

\begin{observation}\label{obs:1dcomplopt}
  For instances of \mscG in 1D for which the underlying graph $G$ is a complete
  graph, there exists a polynomial-time algorithm for computing an optimal scan
  cover.
\end{observation}
\begin{proof}
  Because every scan cover induces a cut cover and $c(G)=\lceil \log_2 n
  \rceil$, it suffices to provide a scan cover with this number of steps.
  To this end, we recursively scan the bipartite graphs induced by two vertex sets when split into halves with respect to $<_L$.
\end{proof}

\section{Two-Dimensional Point Sets}\label{sec:2D}

For two-dimensional point sets, we show that even for bipartite graphs,
it is hard to approximate \mscG better than $\nicefrac{3}{2}$.
Conversely, we present a $\nicefrac{9}{2}$-approximation algorithm for these graphs and apply the gained insights to achieve approximations for $k$-colorable graphs.

\subsection{Bipartite Graphs}

By \cref{thm:noapx}, we cannot hope for a constant-factor approximation for general graphs.
However, bipartite graphs in 1D can be solved in polynomial time.
We show that the added
geometry of 2D makes the \mscG hard to approximate even for bipartite graphs.

\subsubsection{No Approximation Better than 1.5 for Bipartite Graphs in 2D}
As a stepping stone for the geometric case, we establish the following.
\begin{restatable}{lem}{lemmaabstracthardness}
  \label{th:abstract_hardness}
  It is NP-hard to approximate \mscA better than $\nicefrac{3}{2}$ even for bipartite graphs.
\end{restatable}

\begin{proof}
	The proof is based on a reduction from \textsc{Not-All-Equal-3-Sat} where  a satisfying assignment fulfills the property that each clause has a {\tt true} and a {\tt false} literal, i.e., all false or all true is prohibited. 
	The nice feature of this variant is that the negation of a satisfying assignment is also a satisfying assignment.
	
	For every instance $I$ of \textsc{Not-All-Equal-3-Sat}, we construct a graph $G_I$ and a cost function $\alpha$ where each edge pair has a transition cost of $0$, $\phi$, or $2\phi$. Thus, every optimal scan cover has discrete time steps at distance $\phi$.   We show that there exists a scan cover of $(G_I,\alpha)$ with three time steps, i.e., a  scan time of $2\phi$, only if $I$ is a satisfiable instance. Otherwise, every scan cover has at least four steps, i.e., a value of $3\phi$.
	
	We now describe our construction of $G_I$, which is a special variant of a clause-variable-incidence graph.  For an illustration, see \cref{fig:nphardness_abstract_example}.
	\begin{figure}[htb]
		\centering
		\includegraphics[page=2]{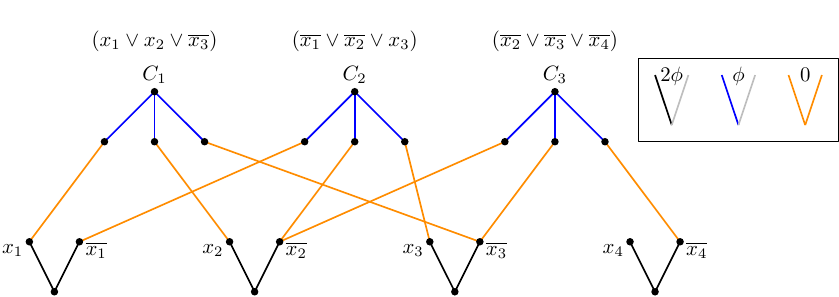}
		\caption{Illustration of the graph $G_I$ for the instance 
			$I=(x_1\lor x_2\lor\overline{x_3})\land(\overline{x_1}\lor\overline{x_2}\lor x_3)\land(\overline{x_1}\lor\overline{x_2}\lor\overline{x_3})$. The edge set consists of clause edges (blue), incidence edges (orange), and variable edges (black).}
		\label{fig:nphardness_abstract_example}
	\end{figure}
	There are four types of vertices and three types of edges:
For every clause $C_i$ of $I$, we introduce a \emph{clause gadget} consisting of a \emph{clause vertex} and  three \emph{entry vertices}, each of which represents one of the literals appearing in the clause. The clause vertex is adjacent to every entry vertex of its gadget by a  \emph{clause edge}.  
For every variable $x_i$ of $I$,  we introduce a \emph{variable vertex} and two \emph{literal vertices}. The variable vertex is adjacent to both literal vertices via a \emph{variable edge}.
Moreover, for every entry vertex, we introduce an \emph{incidence edge} to the literal vertex that it represents. 
	
	We define $\alpha$ as follows: The transition cost for any edge pair is $\phi$ if it contains a clause edge, $2\phi$ if it contains a variable edge, and 0 otherwise. Note that every variable and clause edge are pairwise disjoint; hence this is well-defined.

	We now show that if $I$ is a satisfiable instance of \textsc{Not-All-Equal-3-Sat}, then there exists a scan cover with three time steps: 
	If a literal is set to true, then the variable edge of this literal vertex  is scanned in the first time step and all remaining edges of the literal vertex in the third step. Likewise, if a literal is false, then its variable edge is scanned in the third step, and all other incident edges in the first step.
	
	For each clause we choose one positive and negative literal to be \emph{responsible}, the third literal is \emph{intermediate}. The clause edges are scanned in the first, second, or third step, depending on whether their entry vertex corresponds to a responsible positive literal, an intermediate literal, or a responsible negative literal, respectively. Note that the edge pairs with transition costs of $2\phi$, namely the edges incident to literal vertices, are scanned in the first or third step. Thus, the value of this scan cover is $2\phi$.  For an example, consider \cref{fig:nphardness_abstract_exampleScan}.
	
	\begin{figure}[htb]
		\centering
		\includegraphics[page=3]{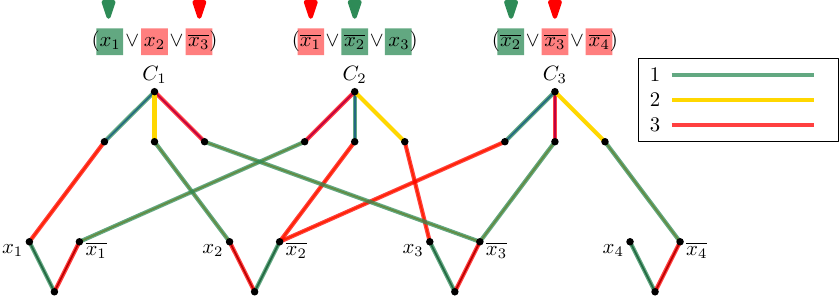}
		\caption{Illustration of a scan cover of the graph $G_I$ as in \cref{fig:nphardness_abstract_example}. Green edges are scanned in the first, yellow in the second, and red edges in the third step.}
		\label{fig:nphardness_abstract_exampleScan}
	\end{figure}
	
	Now, we consider the reverse direction and show that a scan cover with three time steps corresponds to a satisfying assignment of $I$. Because the transition cost of any two edges incident to a literal vertex is $2\phi$, each variable or incidence edge is scanned either in the first or third step.
	Therefore, we may define an assignment of $I$ by setting the literals whose variable edge is scanned in the first time step to true.
	It remains to argue that in this assignment, every clause has a true and false literal.
	Note that the three edges of a clause gadgets, must be scanned at different time steps. Consequently, there exists a clause edge that is scanned in the first time step. Its adjacent incident edge is therefore scanned in the third step. This implies that the variable edge of the literal vertex is also scanned in the first time step and thus set to true. Likewise, the clause gadget in the third step corresponds to a false literal. Consequently, this assignment shows that $I$ is a true-instance of \textsc{Not-All-Equal-3-Sat}.
\end{proof}

We now use \cref{th:abstract_hardness} for showing hardness of bipartite graphs in the geometric version.

\begin{restatable}{theorem}{theoremhardnessbipartite}\label{thm:hardnessBIP}
  Even for bipartite graphs in 2D, a C-approximation for \mscG for any $C < \nicefrac{3}{2}$ implies P = NP.
\end{restatable}
\begin{proof}
	Suppose that there is a $(\nicefrac{3}{2}-\varepsilon)$-approximation for some $\varepsilon>0$.
	For every instance $I$ of \textsc{Not-All-Equal-3-Sat}, we can construct a graph $G_I$ for \mscG in 2D such that it has a scan time of $240^\circ$ if $I$ is satisfiable, and a scan time of at least $360^\circ-\varepsilon$ otherwise.
	We essentially use the same reduction as in the proof of \cref{th:abstract_hardness}.
	It remains to embed the constructed graph $G_I$ in the plane such that the transition costs are reflected by the angle differences.  	The basic idea is to embed $G_I$ on a triangular grid; see \cref{fig:hardness_on_plane_1_5} for some of the gadgets.
	
	\begin{figure}[htb]
		\centering
		\includegraphics[page=3]{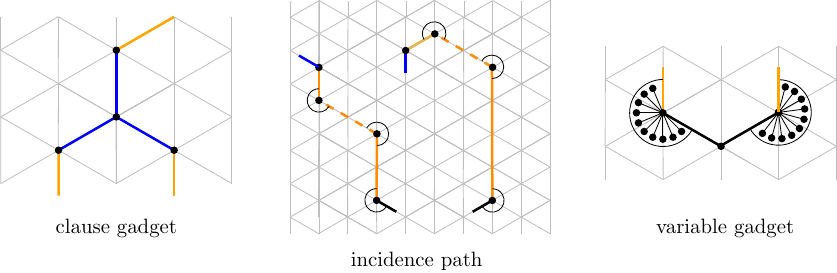}
		\caption{Embedding the graph $G_I$ into the plane by using $\phi=120^\circ$. Additional leaves are added to force the usage of the larger angle of $240^\circ$. The clause and variable gadgets are connected by paths instead of edges (solid and dashed orange edges).}
		\label{fig:hardness_on_plane_1_5}
	\end{figure}

	In particular, we choose $\phi=120^\circ$.
	For each \emph{clause gadget} we create a star on four vertices with $120^\circ$ angles between the edges.
	The incidence edges also leave with $120^\circ$ from the three entry vertices.
	
	The vertices of the \emph{variable gadget} can also easily be embedded in the triangular grid. However, because the smaller angle between any two segments is at most $180^\circ$, we cannot directly construct angles of $240^\circ$.
	Therefore, we insert additional edges and vertices into the $240^\circ$ angle with an angle difference of $\varepsilon$ as illustrated in \cref{fig:hardness_on_plane_1_5}.
	If an incident vertex uses the shorter $120^\circ$ angle, it would still need to cover the additional edges resulting in an overall turning angle of at least $360^\circ-\varepsilon=3\phi-\varepsilon$.
	
	To connect the \emph{clause gadgets} with the \emph{variable gadgets} we now need \emph{incidence paths} instead of incidence  edges.
	We use paths consisting of three edges with angles of $240^\circ$ on the interior vertices.
	A path will propagate the decision by always scanning all odd or all even edges at the same time with a difference of $240^\circ$.
	Thus, the first and the last edge of it are scanned at the same time.
	
	If we allow the points to share the same coordinates, we can position all clause and variable gadgets at the same locations, respectively. This results in a constant number of coordinates.
	
	If all coordinates shall be unique, the gadgets can easily be moved up or down as the \emph{incident paths} can be stretched.
	This replicates the behavior of the original construction except of a tiny angle difference of $\varepsilon$ for the $2\phi$ angles.

	A $(\nicefrac{3}{2}-\varepsilon)$-approximation would now yield for a satisfiable instance a scan time of at most $(\nicefrac{3}{2}-\varepsilon)\cdot240^\circ=360^\circ-\varepsilon\cdot240^\circ$ and decide the satisfiability because an unsatisfiable solution would have a scan time of at least $360^\circ-\varepsilon>360^\circ-\varepsilon\cdot240^\circ$.
	This is a contradiction to the NP-hardness of \textsc{Not-All-Equal-3-Sat}.
\end{proof}


\subsubsection{4.5-Approximation for Bipartite Graphs in 2D}
Conversely, we give absolute and relative performance guarantees for bipartite graphs in 2D.
\begin{theorem}\label{lemma:scan_points_splited_by_line}
  Let $I=(P,E)$ be a bipartite instance of \mscG with vertex classes $P=P_1\cup P_2$.
  Then $I$ has a scan cover of time $360^\circ$.
  Moreover, if $P_1$ and $P_2$ are separated by a line, there is a scan cover of time $180^\circ$.
\end{theorem}

\begin{proof}
  We show that the following strategy yields a scan cover of time $360^\circ$: All points turn in clockwise direction,
  with the points in $P_1$ starting with heading north and the points in $P_2$ with heading south; see \cref{fig:bipartiteA} for an example.
  Note that the connecting line between any point $p_1\in P_1$ and any point $p_2\in P_2$ forms
  alternate angles with the parallel vertical lines through $p_1$ and $p_2$, so both face each other
  when reaching this angle during their rotation; see \cref{fig:bipartiteB}. In the case of separated
  point sets, a rotation of $180^\circ$ suffices to sweep the other set, as illustrated in \cref{fig:bipartiteC}.
\end{proof}

\begin{figure}[htb]
  \centering
  \begin{subfigure}[b]{.3\textwidth}
    \centering
    \includegraphics[page=2]{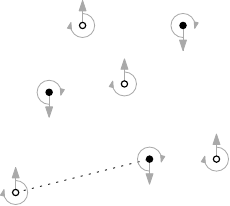}
    \caption{}
    \label{fig:bipartiteA}
  \end{subfigure}
  \hfil
  \begin{subfigure}[b]{.4\textwidth}
    \centering
    \includegraphics[page=3]{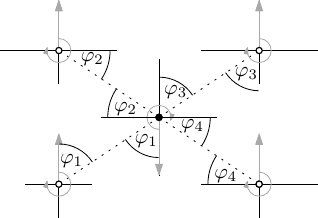}
    \caption{}
    \label{fig:bipartiteB}
  \end{subfigure}
  \hfil
  \begin{subfigure}[b]{.27\textwidth}
    \centering
    \includegraphics[page=4]{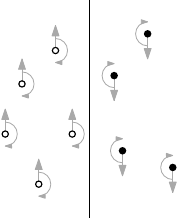}
    \caption{}
    \label{fig:bipartiteC}
  \end{subfigure}
  \caption{\textbf{(a)} The vertices in $P_1$ and $P_2$ rotate clockwise and start by heading north and south, respectively. \textbf{(b)} Due to alternate angles, vertices of different parts of the vertex partition face each other at the same time. \textbf{(c)} When $P_1$ and $P_2$ are separated by a line, a scan time of $180^\circ$ suffices.}
  \label{fig:bipartite}
\end{figure}

\cref{lemma:scan_points_splited_by_line} yields an \emph{absolute} bound for bipartite graphs.
Now we give a constant-factor approximation even for small optimal values.

\begin{theorem}\label{thm:approxBip}
  There is a $4.5$-approximation algorithm for \mscG for bipartite graphs in 2D.
\end{theorem}
\begin{proof}
  Consider an instance $I$ of \mscG in 2D and let $\Lambda$ denote the minimum angle such that for every vertex some $\Lambda$-cone contains all its edges. Clearly, $\Lambda$ is a lower bound on the value $OPT$ of a minimum scan cover of $I$. We use one of two strategies depending on $\Lambda$.

  If $\Lambda\geq 90^\circ$, we use the strategy of \cref{lemma:scan_points_splited_by_line} which yields a scan cover of at most~$360^\circ$ and hence a $4$-approximation.

  If $\Lambda<90^\circ$, we use an adaptive strategy as follows.
  For each vertex, we partition the set of headings $[0, 360^\circ)$ into $2s$ sectors of
  size $\Lambda'=\nicefrac{360^\circ}{2s}$, see \cref{fig:bip_apx_alpha_cone}.
  We choose $s$ maximal (and, thus, $\Lambda'$ minimal) such that $\Lambda'\geq \Lambda$.
  This implies that the edges of every vertex are contained in at most two
  adjacent sectors, see \cref{fig:bip_apx_alpha_cone}. 	Note also that
  $\Lambda'<\nicefrac{3}{2}\Lambda$, because $\Lambda>\nicefrac{360^\circ}{2(s+1)}$ and
  $s\geq 2$.

  \begin{figure}[htb]
    \begin{subfigure}[b]{.45\textwidth}
      \centering
      \includegraphics{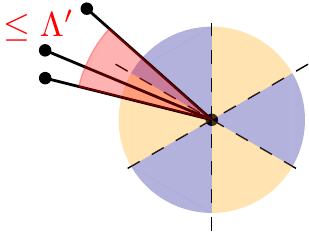}
      \caption{}
      \label{fig:bip_apx_alpha_cone}
    \end{subfigure}
    \hfil
    \begin{subfigure}[b]{.45\textwidth}
      \centering
      \includegraphics{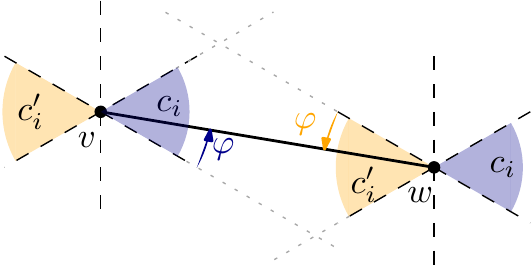}
      \caption{}
      \label{fig:bip_apx_sector_edge}
    \end{subfigure}
    \caption{%
      \textbf{(a)} Dividing the headings into $2s$ sectors with angle $\Lambda'=\nicefrac{360^\circ}{2s}$ by inserting $s$ lines. For every vertex, the incident edges lie in at most two adjacent sectors of size $\Lambda'$, because $ \Lambda \leq \Lambda'$. \\
      \textbf{(b)} An edge $e=vw$ that lies in $c_i$ for $v$, lies in $c'_i=c_{i+s}$ for $w$. If $v$ scans $c_1$ and $w$ scans $c'_1$ (both counterclockwise), they scan $e$ at the same time $\varphi$ due to the alternate angles in the parallelogram.}
  \end{figure}

  Let the sectors be $c_i=[(i\cdot \Lambda',(i+1)\cdot \Lambda')$, for $i=0,\dots,2s-1$. Moreover, $c_i':=c_{i+s\mod 2s}$ is the sector \emph{opposite} of $c_i$.
  Note that an edge $e=vw$ is in the sector~$c_i$ of $v$ if and only if $e$ is in the opposite sector $c'_i$ of $w$, see \cref{fig:bip_apx_sector_edge}.
  Let $C_\text{even}$ be the set of sectors with even indices, $C_\text{odd}$ the one with odd indices, and $C_\text{even}'$ and $C_\text{odd}'$ the set of opposite sectors, respectively.
  Because the incident edges of each vertex are contained in at most two adjacent sectors, every vertex has edges in (at most) one sector of $C_\text{even}$ and one sector of $C_\text{odd}$.

  This allows the following strategy.
  Denote the bipartition of the vertex set by $P=P_1\cup P_2$. In the first phase,
  the vertices in $P_1$ scan the sector with edges in $C_\text{even}$
  in clockwise direction, while the vertices in $P_2$ scan the sector in
  $C'_\text{even}$. In the second phase,
  vertices in $P_1$ scan the sector with edges in $C_\text{odd}$
  in counterclockwise direction, while the vertices in $P_2$ scan the sector
  in $C_\text{odd}'$.

\begin{figure}[b]
	\centering
	\includegraphics{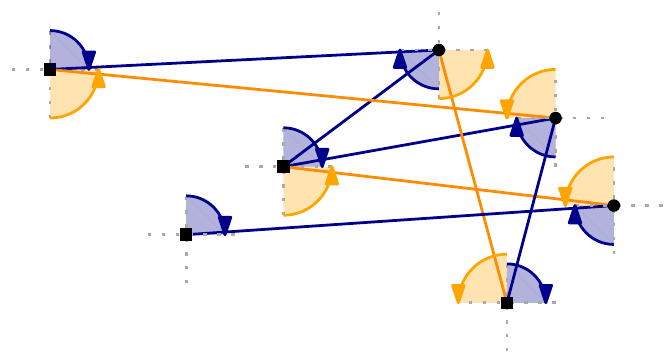}
	\caption{An example with $\Lambda'=90^\circ$. $P_1$ and $P_2$ are indicated by squares and circles, respectively.
		The blue sectors are scanned in the first scan phase; the orange sectors in the second. Each scan phase and the turning phase costs $\Lambda'$.}
	\label{fig:bip_approx_example}
\end{figure}

\cref{fig:bip_approx_example} depicts an example scan cover.
  As in \cref{lemma:scan_points_splited_by_line}, every edge is scanned in the first or second
  phase  due to the alternate angles. Clearly, each scan phase needs
  $\Lambda'$. Between the two scan phases, every vertex $v$ needs to turn to
  change its heading from the end heading of the first scan phase to the start
  heading of the second. Because both sectors of $v$ are incident and due to the
  reversed direction, the turning angle is at most $\Lambda'$; in particular,
  either the end heading of the first sector is contained in the boundary of the second sector or the two start
  headings of both phases coincide.
  The resulting scan time is $3\Lambda'\leq 3\cdot \nicefrac{3}{2}\Lambda \leq 4.5\cdot  \text{OPT}$.
\end{proof}

\subsection{Graphs with Bounded Chromatic Number}
Like in 1D, the value of a minimum scan cover in 2D has a strong relation to the chromatic number.
More specifically, we show that the optimal scan time lies in $\Theta(\log_2 \chi(G))$ and that for a given coloring of the graph $G$ with $\chi(G)^c$ colors, we can provide an $O(c)$-approximation.

\begin{lem}\label{lem:lowerBoundChromatic}
	Let $I=(P,E)$ be an instance of \mscG in $\mathbb R^d, d>1$. If $I$ has a scan cover of length $T>0$, then
	$G$ has a cut cover of size $d\cdot\lceil \frac{T}{90^\circ}\rceil $, i.e., $c(G)\leq d\cdot\lceil \frac{T}{90^\circ} \rceil$.
\end{lem}

\begin{proof}
  Partition the scan cover into $\lceil \frac{T}{90^\circ}\rceil$ intervals of length at most $90^\circ$.
  For each interval~$i$, we consider the set of edges that are scanned within this interval, inducing a graph $G_i$.
  We show that each $G_i$ is $2^d$-partite. Because $c(G_i)=\lceil \log_2(\chi(G_i))\rceil\leq d$, this implies the claim.

  We first consider the case $d=2$.
  We classify the points of $P$ into four sets, depending on their turning behavior within the interval $i$.
  Each point has a quadrant $[0, 90^\circ)$, $[90^\circ, 180^\circ)$, $[180^\circ, 270^\circ)$, or $[270^\circ, 360^\circ)$ to which it is heading at the time $45^\circ$; we assign each point this quadrant. 
  Note that every point can only leave its assigned quadrant by less than $\pm 45^\circ$.
  Two points that are assigned to the same quadrant are independent in $G_i$:
When their edge is scanned, the headings of the two points have to be opposite,
i.e., they differ by exactly $180^\circ$. Thus, the only case in which two point
headings could differ by $180^\circ$ is if one leaves its quadrant by
$45^\circ$ in clockwise and the other by $45^\circ$ in counterclockwise
direction.
 However, in this case, the points would not have been assigned to the same half-open sector.
For an illustration consider \cref{fig:2d:chromatic_lb_sector_assignment}.

   \begin{figure}[htb]
 	\centering
 	\includegraphics[page=2]{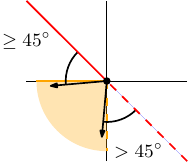}
 	\caption{Every vertex is assigned to the (orange) sector to which it is heading at time $45^\circ$. The boundary of the reachable headings within $90^\circ$ is shown in red. Since the sector is half-open, two vertices assigned to the same sector cannot reach opposing headings and thus cannot scan their edge.}
 	\label{fig:2d:chromatic_lb_sector_assignment}
 \end{figure}

For $d\geq3$, the idea is analogous.
To simplify  the argument we choose a coordinate system, i.e., an orthonormal
basis (ONB), such that, at time $45^\circ$, no point heads in a direction that
lies on a lower-dimensional subspace spanned by the basis vectors.  Let $B$
denote the set of all potential basis vectors in~$\mathbb R^d$, i.e.,
$B=\mathbb R^d$. For every point, we delete the line spanned by its heading $h$
at time~$45^\circ$ and the $(d-1)$-dimensional subspace orthogonal to $h$ from
$B$. The remaining set $B'$ is a $d$-dimensional space minus a
finite number of lower-dimensional subspaces. It follows by induction that $B'$
contains an ONB.

%

The points of $P$ are partitioned into $2^d$ different sets, depending on the orthant in which they are contained at time $45^\circ$.
Note that if two point headings of the same orthant differ by $180^\circ$, their angle difference at time $45^\circ$ is $90^\circ$, i.e., they lie on a lower dimensional subspace spanned by the basis vectors. This is a contradiction to the choice of the ONB.
\end{proof}

Because $c(G)=\lceil \log_2 \chi(G)\rceil$, \cref{lem:lowerBoundChromatic} has the following implication.
\begin{lem}\label{thm:lowerBoundKcolorable}
  Every instance $I$ of \mscG in $\mathbb R^d$ needs a scan time $T$ of at least $\Omega(\log_2 \chi(G_I))$, with~$G_I$ denoting the underlying graph of $I$.
  More precisely, $T\geq \frac{\lceil \log_2 \chi(G) \rceil -d}{d}\cdot 90^\circ$.
\end{lem}
\begin{proof}
  Because $c(G)=\lceil \log_2 \chi(G)\rceil$, \cref{lem:lowerBoundChromatic} implies that $\lceil \log_2 \chi(G)\rceil\leq d\lceil \frac{T}{90^\circ}\rceil$. In particular, it holds that 
  $
  \lceil \log_2 \chi(G) \rceil \leq d\frac{T}{90^\circ}  +d
  \iff \frac{\lceil \log_2 \chi(G) \rceil -d}{d}\cdot 90^\circ\leq T.
  $
\end{proof}

These insights have the following implications.
\begin{theorem}\label{cor:approxCol}
  For instances of \mscG in 2D with a $k$-coloring of the
  graph~$G=(V,E)$, such that $k\leq   \chi(G)^c$ for some function $c$,
there is an $O(c)$-approximation. \end{theorem}

\begin{proof}
  Partition $G$ into $\lceil \log_2 k\rceil$ bipartite graphs $G_i$. By
  \cref{thm:approxBip}, each $G_i$ can be scanned in time $\beta\cdot OPT_i$,
  with $\beta=4.5$ and $OPT_i$ denoting the optimum of the instance induced
  by~$G_i$. Clearly, $OPT_i\leq OPT$, and turning from the last
  scan of one bipartite graph to the next takes at most a time of $OPT$.
  Hence, this
  scan cover needs a time of at most $(\beta+1)OPT\lceil \log_2 k\rceil$.

  If $\chi(G)\leq 4$, then
  $O(\lceil \log_2 k\rceil)
  \leq O(\lceil c\cdot \log_2( \chi(G))\rceil)
  \leq O(\lceil c\cdot \log_2(4)\rceil)
  \leq O(\lceil 2c\rceil)\in O(c)$.

  If $\chi(G)\geq 5$, then \cref{thm:lowerBoundKcolorable} ensures that $OPT\geq\Omega(\log_2 (\chi (G)))>0$.
  Therefore, the performance guarantee is in $O\left(\frac{\log_2 k}{\log_2 (\chi
  (G))}\right)= O\left(\frac{c \log_2( \chi (G))}{\log_2 (\chi (G))}\right)=O(c)$.
\end{proof}

As a direct implication of \cref{cor:approxCol}, we get a spectrum of approximation algorithms for interesting special cases.
\begin{corollary}\label{cor:specialCases}
  \mscG in 2D allows the following approximation factors.
  \begin{enumerate}
    \item $O(\log_2 n)$ for all graphs. Furthermore, the minimum scan time lies in $\Theta(\log_2 \chi (G))$.
    \item $O(1)$ for planar graphs.
    \item $O(\log_2 d)$ for $d$-degenerate graphs.
    \item $O(1)$ for graphs of bounded treewidth.
    \item $O(1)$ for complete graphs.
  \end{enumerate}
\end{corollary}
The following bound shows a refined approximation for complete graphs.

\begin{corollary}\label{cor:complete2D}
  Consider the \mscG for complete graphs with $n$ vertices in 2D. There is a $c$-approximation algorithm with $c\to 6$ for $n\to \infty$.
\end{corollary}
\begin{proof}
  We may  assume  without loss of generality that $n>4$. By \cref{thm:lowerBoundKcolorable}, the minimum scan time is at least $(\lceil\log_2(n)\rceil-2)\cdot 45^\circ>0$.
  For the upper bound, we partition the point set recursively into $\lceil\log_2(n)\rceil$ bipartite graphs by lines (alternating horizontal and vertical).
  Hence, \cref{lemma:scan_points_splited_by_line} allows us to scan each bipartite graph within $180^\circ$.
  The transition between two scan phases is at most~$90^\circ$.
  Therefore,  the scan time is upper bounded by $\lceil\log_2(n)\rceil 180^\circ+(\lceil\log_2(n)\rceil-1)90^\circ$.
  This yields a performance guarantee of
  \[
    \frac{270^\circ(\lceil\log_2(n)\rceil-2)+450^\circ}{45^\circ(\lceil\log_2(n)\rceil-2)}
    =6+\frac{10}{(\lceil\log_2(n)\rceil-2)}.\qedhere
  \]
\end{proof}
The factor in \cref{cor:complete2D} is $c\leq 8$ when  $n\geq 2^7$ and $c\leq 7$ when $n\geq 2^{12}$.

\section{Three-Dimensional Point Sets and Abstract MSC}\label{sec:3D}

In the following, we observe that \mscA generalizes the Path-TSP.

\begin{observation}\label{obs:starsTSP}
  Let $G=(V,E)$ be a star on $n+1$ vertices with center $v$ and $\alpha$ a metric transition cost function on $E\times E$.
  Then, an \mscA of $(G,\alpha)$ corresponds to a TSP-path of the complete graph on $V\setminus\{v\}$ with metric cost $c(u_1, u_2)=\alpha(vu_1, vu_2)$ and vice versa.
\end{observation}

\cref{obs:starsTSP} has two immediate consequences. Firstly, because the metric Path-TSP is NP-hard, it follows that

\begin{observation}
  \label{th:hard:amsc_star}
  \mscA is NP-hard even for stars.
\end{observation}

Secondly, the $1.5$-approximation for metric Path-TSP by Zenklusen \cite{Zenklusen_pathTSP} can be applied.
\begin{observation}
  There exists a $1.5$-approximation algorithm for \mscA for stars.
\end{observation}

In contrast to 1D and 2D, we show that the chromatic number does not provide
an upper bound for \mscG in 3D and \mscA.

\begin{observation}\label{obs:noboundbychrom}
  There are instances of \mscG in $3D$ with $\chi(G)=2$ that need at least $\Omega(\sqrt{n})$.
  There are instances of \mscG in $nD$ with $\chi(G)=2$ that need at least $\Omega(n)$.
\end{observation}

\begin{proof}
  For the first claim consider a geodesic triangular grid on a sphere and embed a star graph such that its leaves are grid points and the center of the star lies in the center of the sphere.
  The $\Omega(\sqrt{n})$ can be achieved by increasing the resolution of the grid by subdivision, see \cref{fig:3d_scaling}: While the minimum turn cost between two consecutive edges approximately halves, the number of vertices roughly quadruples, doubling the overall costs that  is lower bounded by $(n-1)\cdot l$ if $l$ is the minimum edge length.

  \begin{figure}[htb]
    \centering
    \includegraphics[scale=1]{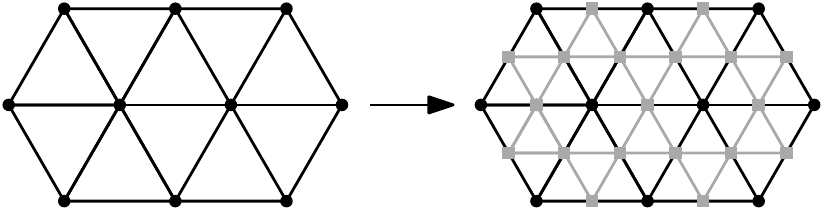}
    \caption{The geodesic grid is refined by subdividing the edges. Because a triangulation with $n$ vertices has $3n-6$ edges,  subdividing roughly quadruples the number of vertices.}
    \label{fig:3d_scaling}
  \end{figure}

  The second claim follows from considering a star on $n$ vertices for which each leaf is placed on a different coordinate
  axis. Therefore, the turn cost between any two edges is~$90^\circ$ and it takes $90^\circ n$ to scan the graph.
\end{proof}

The approximation technique for bipartite graphs in 2D relies on alternate angles and fails for $3D$ or \mscA.
Nevertheless, we provide a $2.5$-approximation for trees.

\begin{theorem}\label{thm:3DTree}
  There exists a $2.5$-approximation algorithm for \mscA for trees.
\end{theorem}
\begin{proof}
  Let $I=(G,\alpha)$ be an instance of $\mscA$ for which $G$ is a tree, and let $OPT$ be the minimum scan time of $I$.
  For every vertex~$v$, we approximate an ordering of minimum cost over all its incident edges $E_v$. Let $N(v)$ denote the set of neighbors of $v$. By \cref{obs:starsTSP} such an ordering corresponds to a TSP-path. Consequently, we may use the $1.5$-approximation
  algorithm for metric Path-TSP by Zenklusen \cite{Zenklusen_pathTSP}. Moreover,
  we enhance the edge ordering to a cyclic ordering by inserting an edge from the
  last to the first edge; because the cost function is metric, the cost of the
  additional edge is upper bounded by the minimum cost ordering of the incident
  edges. Therefore, the scan time $\ell_v$ of the computed cyclic edge ordering
  of $v$ is at most $\ell_v\leq (1.5 +1)OPT$.

  We construct a scan cover as follows:
  Every vertex follows its cyclic edge ordering.
  The start headings of the vertices are chosen such that the scan time of each
  edge $e=uv$ is synchronized at the vertices $u$ and $v$. To this end, we choose
  some vertex $r$ as the root and denote the parent of each vertex $v$ by
  $par(v)$ in the tree $G$ with respect to the root $r$.
  We scan the edges of $r$ according to the cyclic edge orderings by starting with any heading, see also \cref{fig:3d_tree_apx_example}.

  \begin{figure}[htb]
    \centering
    \includegraphics[page=2]{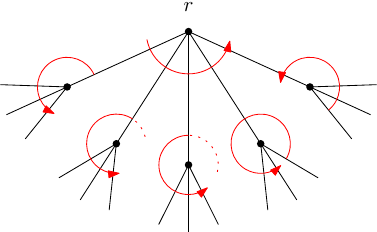}
    \caption{Root $r$ can choose its schedule. The (cyclic) schedules of the children are synchronized with the timing of its parent. Because the graph is a tree, there are no cyclic dependencies.}
    \label{fig:3d_tree_apx_example}
  \end{figure}
  We then determine the start headings by a tree traversal (e.g., DFS or BFS):
  Let $v$ be a vertex whose start heading has to be determined, and assume the
  start heading of $u:=par(v)$ is already fixed. When $uv$ is scanned
  at time $t$ for $u$, then the cyclic ordering of $E_v$ is shifted, so that $v$
  sees $uv$ also at time $t$. If this time lies between two scans, we
  simply start at the next incident edge and let the vertex wait for the
  appropriate time.
  Because all vertices start at the same time, the resulting scan cover has a scan time of at most $\max_v \ell_v\leq 2.5 \cdot OPT$.
\end{proof}

\cref{thm:3DTree} allows an approximation algorithm in terms of the arboricity
of the underlying graph. Recall that the \emph{arboricity} of a graph denotes
the minimum number of forests into which its edges can be partitioned.

\begin{theorem}\label{thm:amscArboricityApprox}
  There is a $3.5A$-approximation for the \mscA for graphs of arboricity $A$.
\end{theorem}
\begin{proof}
  We compute a decomposition into $A$ forests in polynomial time
  \cite{Gabow1992}.
  To obtain a scan cover we use
  the approximation algorithm of \cref{thm:3DTree} for each forest and
  concatenate the resulting scan covers in any order. Because the transition cost
  between any two forests is upper bounded by the minimum scan time $OPT$, the
  resulting scan cover has time of at most $(2.5+1)OPT\cdot A$. Consequently, we obtain a $3.5A$-approximation.
\end{proof}

\section{Conclusion and Open Problems}\label{sec:Conclusion}

We have presented several algorithmic results for the abstract and geometric
version of the minimum scan cover problem with a metric angular cost function,
which has strong connections to  the chromatic number.

There is a spectrum of interesting directions for future work.
How can we make use of methods for solving graph coloring problems to compute practical solutions to real-world instances?
In many scenarios, this will involve considering satellites on different trajectories, in the presence
of a large obstacle: the earth.
 This gives rise to a variety of generalizations, such as the presence of time windows for possible communication.
Other variations of both practical and theoretical interest arise
from considering objective functions such as minimizing the total energy of all satellites or minimizing the 
maximum energy per satellite.


\bibliography{biblio}

\clearpage
\newpage
\appendix

\end{document}